\documentclass[11pt]{article}

\usepackage{amssymb,amsmath}
\usepackage{amsthm}
\usepackage{color,pdfpages,dsfont}
\usepackage{soul}

\newtheorem{theorem}{Theorem}[section]
\newtheorem{lemma}[theorem]{Lemma}

\newtheorem{corollary}[theorem]{Corollary}

\newcommand{\marginlabel}[1]%
{\mbox{}\marginpar{\it{\raggedleft\hspace{0pt}#1}}}

\newlength{\pgmtab}  
\newcommand{\x}{\mbox{\boldmath $x$}}

\def\p{\mbox{\boldmath $p$}}

\def\x{\mbox{\boldmath $x$}}

\def\uij{u_{ij}}

\newcommand{\eat}[1]{}

\newcommand{\xij}{ x_{ij}}

\renewcommand{\Pi}{\mathcal{P}}

\newlength{\lpskip}
\newcommand{\market}{\mathcal{M}}
\newcommand{\cij}{c_{ij}}

\begin{document}
\title{Market Equilibrium with Transaction Costs}

\author{Sourav Chakraborty\thanks{Technion.  \texttt{sourav@cs.technion.ac.il}} \and  Nikhil R Devanur\thanks{Microsoft Research. \texttt{nikdev@microsoft.com}} \and Chinmay Karande\thanks{Georgia Institute of Technology. \texttt{ckarande@cc.gatech.edu}. Part of the work done when
the author was at Microsoft Research, Redmond.}}
\date{}

\maketitle

\begin{abstract}
Identical products being sold at different prices in different
locations is a common phenomenon. Price differences might occur due to various reasons such as shipping costs, trade restrictions and price discrimination. To model such scenarios, we supplement the classical Fisher model of a market by introducing {\em transaction costs}. For every buyer $i$ and every good $j$, there is a transaction cost of $\cij$; if the price of good $j$ is $p_j$, then the cost to the buyer $i$ {\em per unit} of $j$ is $p_j + \cij$. This allows the same good to be sold at different (effective) prices to different buyers.

We provide a combinatorial algorithm that computes $\epsilon$-approximate equilibrium prices and allocations in
$O\left(\frac{1}{\epsilon}(n+\log{m})mn\log(B/\epsilon)\right)$
operations - where $m$ is the number goods, $n$ is the number of buyers and $B$ is the sum of the budgets of all the buyers. 
\end{abstract}

\section{Introduction}

Identical products being sold at different prices in different locations is a common phenomenon.
Price differences might occur due to different reasons such as
\begin{itemize}
\item Shipping costs. Oranges produced in Florida are cheaper in Florida
than they are in Alaska, for example.
\item Trade restrictions. A seller with access to a wider market might sustain a higher
price than one that does not. Kakade \textit{et al} \cite{kakade-econ}
considered a model called graphical economies to study how price differences occur
due to trade agreements in international trade. Note that this is different from shipping costs since two countries are either allowed to trade (in which case they pay the same price) or not.
\item Price discrimination. A good might be priced differently for
different people based on their respective ability to pay. For example,
conference registration fees are typically lower for students than for professors.
\end{itemize}

In order to capture all these scenarios,
we supplement the classical Fisher model of a market
(see below for a formal definition) by introducing {\em transaction costs}.
For every buyer $i$ and every good $j$, there is a transaction cost of $\cij$;
if the price of good $j$ is $p_j$, then the cost to the buyer $i$ {\em per unit} of $j$ is $p_j + \cij$.
This allows the same good to be sold at different (effective) prices to different buyers.
Note that apart from non-negativity, the transaction costs are not restricted in any way and in particular, do not
have to satisfy the triangle inequality.

The scenarios mentioned earlier can all be modeled as follows:
Shipping costs are most naturally modeled as transaction costs.
We can model the trade restrictions using costs that are either  0 or $\infty$.
Our model can also be used to incorporate price discrimination into the equilibrium
computation by using the transaction costs to  specify buyer specific reserve prices.

\noindent \textbf{Alternative Models - Exogenous \textit{Vs} Endogenous Price Differentiation}: In the section on related work, we will discuss other models of price differentiation that choose an approach similar to ours, \textit{i.e.} to augment the market model with extra costs specified for the buyer. A fundamentally different alternative that does not involve such costs, especially when the price difference is due to shipping costs, is to consider `shipping' or `transaction' itself as another good in the market. This special good has a given capacity and the equilibrium price also determines the `price' of shipping, \textit{i.e.}, the transaction costs. A buyer now derives utility if she gets both a good and the corresponding transaction good.  (Each buyer could need the two in a different ratio.)   There are two reasons to choose our approach over the one in which transaction costs are determined by the market: (a) In many cases the transaction costs are {\em exogenously} specified, and one cannot choose the above model.  (b) Even if one does have a choice, our results indicate that there are faster algorithms to compute the equilibrium in the model with given transaction costs than the one with given capacities.  Thus computational considerations might induce one choose the model with given transaction costs over the other one. 

An important offshoot of the algorithmic study is the definition of new economic models, for example, the spending constraint model \cite{Vaz-adwords} was motivated mainly by computational considerations, but has been found to have interesting economic properties as well. Computational considerations was one of the main motivations that led to the definition of our model as well.

\subsection*{Fisher's Market Model with Transaction Costs} In Fisher's model, a market $\market$ has $n$ buyers and $m$ divisible goods. Every buyer $i$ has budget $B_i$. We consider \emph{linear} utility functions, \textit{i.e.}, the utility of a buyer $i$ on obtaining a bundle of goods $\x_i = (x_{i1},x_{i2}, \ldots)$ is $\sum_j \uij \xij$ where $\uij$ are given constants. Each good has an available supply of one unit (which is without loss of generality). In addition to its price, a buyer also pays a transaction cost $\cij$ per unit of good $j$. The allocation bundle for buyer $i$ is a vector $\x_i$ such that $x_{ij}$ denotes the amount of good $j$ allocated to buyer $i$.
A price vector $\p$ is an equilibrium of $\market$ if there exists allocations $\x_i$  such that
\begin{itemize}
\item  $\x_i$ maximizes the utility of $i$ among all bundles that satisfy the budget constraint, i.e.
$$\x_i \in \arg \max_{\mbox{\boldmath $y$}_i} \left\{ \sum_j \uij y_{ij} : \sum_j (p_j+\cij) y_{ij} \leq B_i \right\},$$
\item and every good is fully allocated or is priced at zero, i.e. $\forall j$, either $\sum_i \xij = 1$ or $p_j =0$.
\end{itemize}

\noindent\textbf{Remark}: A more general model is that of Arrow-Debreu, in which the endowments of buyers are goods, and the budget constraint is defined by the income obtained by the buyer by selling his goods at the given prices. With the introduction of transaction costs, the money is not conserved. For this reason transaction costs are not so natural for the Arrow-Debreu model.

\noindent\textbf{Characterization of Market Equilibrium}: We now characterize the equilibrium prices and allocations in our model. In our model, the ratio $u_{ij}/(p_j + c_{ij})$ denotes the amount of utility gained by buyer $i$ through one dollar spent on good $j$. At given prices, a bundle of goods that maximizes the total utility of a buyer contains only goods that maximize this ratio. Let $\alpha_i = \max_j{u_{ij}/(p_j+c_{ij})}$ be the bang-per-buck of buyer $i$ at given prices. We will call the set $$D_i = \{\, j \,|\, u_{ij} = \alpha_i(p_j+c_{ij})\,\}$$ the demand set of buyer $i$. Hence, $x_{ij} > 0\ \Rightarrow\ j \in D_i$. The conditions characterizing these equilibrium prices and allocations appear in table A below.

An $\epsilon$-approximate market equilibrium is characterized by relaxing the market clearing condition (Equation \eqref{eq1}) and optimal allocation condition (Equation \eqref{eq2}). Refer to equations \eqref{eq3} and \eqref{eq4} in table B.\\

\begin{minipage}[h]{0.45\textwidth}
\center{A: Market Equilibrium}
\begin{small}
\begin{eqnarray}
\label{eq7} \forall i & & \sum_{j}{(p_j + c_{ij})x_{ij}} = B_i\\
\label{eq8} \forall j & & \sum_{i}{x_{ij}} \leq 1\\
\label{eq1}\forall j & & p_j > 0\ \Rightarrow\ \sum_{i}{x_{ij}} = 1\\
\label{eq2}\forall i, j & & x_{ij} > 0\ \Rightarrow\ u_{ij} = \alpha_i(p_j+c_{ij})
\end{eqnarray}
\end{small}
\end{minipage}
\begin{minipage}[h]{0.05\textwidth}
\tiny{.}
\end{minipage}
\begin{minipage}[h]{0.50\textwidth}
\center{B: $\epsilon$-Approximate Market Equilibrium}
\begin{small}
\begin{eqnarray}
\label{eq5}& & \sum_{j}{(p_j + c_{ij})x_{ij}} = B_i\\
\label{eq6}& & \sum_{i}{x_{ij}} \leq 1\\
\label{eq3}& & p_j > \epsilon\ \Rightarrow\ \sum_{i}{x_{ij}} \geq 1/(1+\epsilon)\\
\label{eq4}& & x_{ij} > 0\ \Rightarrow\ u_{ij} \geq \alpha_i(p_j+c_{ij})/(1+\epsilon)
\end{eqnarray}
\end{small}
\end{minipage}

The relaxation of exact equilibrium conditions can be achieved in other ways. For example, \cite{GK04} use a definition of $\epsilon$-approximate market equilibrium that relaxes the budget constraints. Our algorithm can be easily adapted to this definition by simple modifications to the termination conditions.\\

\subsection*{Our Results and Methods}
We consider  a Fisher market with {\em linear utilities} and arbitrary transaction costs:  Buyer $i$ incurs a transaction cost of $c_{ij}$ per unit of good $j$.
Our main result is a combinatorial algorithm that computes $\epsilon$-approximate equilibrium prices and allocations in
$O\left(\frac{1}{\epsilon}(n+\log{m})mn\log(B/\epsilon)\right)$
operations - where $m$ is the number goods, $n$ is the number of buyers and $B$ is the sum of the budgets of all the buyers.
This algorithm is a generalization of the auction algorithm of Garg and Kapoor \cite{GK04} to our model with the transaction costs. This generalization is not straight forward; the presence of transaction costs introduces new challenges. We now outline some of these difficulties and our approach to solving them.

\vspace{1mm}

\noindent - Even the existence of an equilibrium in our model does not follow directly from any classical results. One also notices other differences right away: equilibrium prices might be irrational numbers, in contrast to the traditional model where they are guaranteed to be rational. An example where the equilibrium prices are irrational is presented in Appendix \ref{app1}.

\vspace{1mm}

\noindent - The term `auction algorithm' is used to describe ascending price algorithms (such as the one in \cite{GK04}) which maintain a feasible allocation at all times. The algorithm makes progress by revoking a portion of goods currently assigned to a buyer and reallocating it to another buyer offering a higher price. An easy monotonicity property that all variants of the auction algorithm use crucially is
the fact that the total surplus (unspent money of the buyers) decreases throughout the algorithm.
However one cannot have such monotonically decreasing surplus in the presence of transaction costs.
This is because even though the prices are increasing, a good may be reallocated to a buyer with a
lower transaction cost and thus the total money spent by the buyers decreases.
We get around this difficulty by analyzing the running time in a way that does not rely on this property.

\vspace{1mm}

\noindent - Another important property that holds in the traditional model is that an increase in the price of a good has the same effect on
the {\em bang-per-buck} of that good for any buyer. Because of the transaction costs, this property is no longer true and hence, a key lemma in the analysis of the auction algorithm of \cite{GK04},
that a certain directed graph is acyclic, does not hold. We provide a counter-example in Appendix \ref{app3}. We also show that due to this, a fully distributed version of the auction algorithm and in particular the algorithm from \cite{GK04} does not even converge. We get around this by designing a way to reallocate the goods in a cycle that guarantees progress.
This process also leads to an increase in the running time if analyzed in the naive way, and the eventual running time is obtained by a more careful, amortized analysis. In spite of these difficulties, we match the running time of the auction algorithm in \cite{GK04}.

\noindent - Our method of reallocating goods is similar in spirit to the path auctions used by \cite{GK06}. Again however, the properties of monotonic decrease in surplus and acyclicness of the demand graph hold in their case, whereas these properties cease to exist when transaction costs are introduced.
\subsection*{Related Work}

\label{related}

The computation of economic and game theoretic equilibria has been an active area of research over the past decade. Hardness results \cite{DGP-3nash,DGP06,CD-nash,cod-ppad,CDDT09} and algorithmic results \cite{DPS,DPSV08,Jain,GK04,Vaz-adwords,DV-sc,CPV05,cod-taton,JV06,ye-linear,DK08} have been delineating the boundary between what is efficiently computable and what is not.
Recently, there has been a lot of interest in analyzing the convergence of local, distributed processes \cite{WZ07,CF08}.

Convex programming has been one of the main tools in designing algorithms for market equilibrium.
A simple modification of the convex program introduced by \cite{jgarg08,Dev09} captures the equilibria of our problem as its optimal solution. (This is presented in Appendix \ref{convex}.) This proves {\em existence and uniqueness} of equilibria.
It also implies that the ellipsoid algorithm can be used to get a polynomial time algorithm
to compute the equilibrium\footnote{Since the equilibrium could be irrational, the ellipsoid algorithm
would compute an equilibrium with precision $\delta$ in time proportional to $\log(1/\delta)$.}.  
The auction algorithm is combinatorial, runs faster
and provides a simple alternative that can be implemented efficiently in practice. 
Also, the auction algorithm is more amenable to heuristical optimizations and modifications, such as to improve the running time on a specific instance class, or to handle situations in which one already has an equilibrium, an additional buyer is introduced into the market and we need to compute the new equilibrium.
 It is not clear if one can construct an interior point algorithm to solve the convex program. 
 Ye \cite{ye-linear} gave one such algorithm for the Eisenberg-Gale convex program.

Devanur et al. \cite{DPSV08} gave a combinatorial algorithm based on the primal-dual schema to compute an exact equilibrium in the traditional model. We can generalize this algorithm to incorporate transaction costs, but the best running time we can prove is exponential. We defer the description of this algorithm to the full version of the paper.
A strongly polynomial time algorithm for the Fisher linear market was given by Orlin \cite{Orlin10}; 
it does not seem like his ideas can be adapted directly to our setting.

Chen, Ghosh and Vassilvitskii \cite{CGV08} study a model similar to ours,
in the setting of profit-maximizing envy-free pricing (for a single commodity but at different locations)
and show that adding transaction costs that form a metric makes the algorithmic problem {\em easier}.
In contrast, our model is clearly a generalization of the traditional Fisher model, and so the
problem is only {\em harder}.

The transaction costs in our model are independent of the prices.
An alternate model is to let the  transaction cost be
a fixed fraction of the price. Such costs can be interpreted as {\em taxes}
and have been studied by Codenotti et al. \cite{CRV06}.
Taxes could be uniform, that is, depend only on the good, or non-uniform,
that is, depend on the good and the buyer. In the Fisher's model, all our results can be extended with minimum  modifications to the setting where both per-dollar taxes and per-unit transaction costs are present in the market.

\subsection*{Extensions and Open Problems}

All of our results can be easily extended to quasi-linear utilities, that is,
the buyers have utility for money as well, which is normalized to 1.
So the utility of the bundle $\x_i$ is $\sum_j (\uij - p_j) \xij $.
Extending the results to other common utility functions is an  open problem.
In particular, Garg, Kapoor and Vazirani \cite{GKV} extend the auction algorithm to 
separable weak gross substitute utilities. The potential function they use is the 
total surplus, and we don't know a combinatorial bound on the number of events in their algorithm. 
As mentioned earlier, this potential function cannot be used in the presence of transaction costs, 
and therein lies the difficulty in extending our results to this case. 

The auction algorithm for the traditional models can be made to be
distributed and even asynchronous, with a small increase in the running time.
We show that a similar distributed/asynchronous version of the algorithm
may not converge in the presence of transaction costs.
An interesting open question is if there is some other asynchronous/distributed algorithm
that also converges fast. In particular, is there a tattonnement process that converges fast
(like in \cite{CF08})?\\

\noindent\textbf{Outline}: The rest of this paper is structured as follows:  We provide an overview, followed by the details of our algorithm in Section \ref{section:algo}. Section \ref{sec:analysis} has the proof of the correctness of the algorithm and the bound on its running time.

\section{Algorithm}
\label{section:algo}

\begin{theorem}
\label{thm.main}
There exists an algorithm that finds $\epsilon$-approximate equilibrium prices and allocations in\\ $O\left(\frac{1}{\epsilon}(n+\log{m})mn\log(B/\epsilon)\right)$ operations where $B = (1+\epsilon)\sum_{i}B_i$.
\end{theorem}

\subsubsection*{Overview}

Our algorithm maintains a set of prices and allocations and modifies them progressively. To initialize, we set all the prices $p_j = \epsilon$ and all the allocations are empty. The algorithm is organized in rounds. At the end of each round, we raise the price of one good by a multiplicative factor of $1+\epsilon$. Any allocations made before the price raise continue to be charged at the earlier, lower price. Therefore at any point in the algorithm, a good may be allocated to buyers at two different prices, $p_j$ and $p_j/(1+\epsilon)$.
During a round, we take a good away from a buyer at the lower price and allocate it to a buyer (possibly the same buyer) at the current, higher price. We find a sequence of such reallocations such that we find a buyer with positive surplus and a good in her demand set such that all of that good is allocated at the current price. When we find such a buyer-good pair, we increase the price of that good and end the round. The algorithm terminates when the budgets of all the buyers are exhausted.

Following invariants are maintained throughout the algorithm:

  \indent\indent I1: Buyers have non-negative surplus i.e. no buyer exceeds her budget.

    \indent\indent I2: All prices are at least $\epsilon$.

    \indent\indent I3: Every good is either priced $\epsilon$ or is fully allocated.

    \indent\indent I4: Any good $j$ allocated to a buyer $i$ must be approximately most desirable. (As in Equation \eqref{eq4})

    \indent\indent I5: A good $j$ is allocated at price either $p_j$ or $p_j/(1+\epsilon)$ where $p_j$ is the current price.

Invariant I3 is a tighter version of equation \eqref{eq3}. We maintain I3 and I5 until the end of the algorithm whence we merge the two price tiers. This may lead to some goods being undersold, but we prove that equation \eqref{eq3} still holds. Also note that invariant I4 holds for any allocations, whether at the higher or lower price tier. Unless mentioned otherwise, the statements of all the lemmas that follow are constrained to maintain these invariants.

We now present the details of our algorithm. Each round consists of roughly two parts: 1) We construct a \emph{demand graph} $G$ on the set of buyers and 2) We perform multiple iterations of a reallocation procedure  - which we call a \emph{transfer walk}. At the end of each round, we increment the price of some good. The sequence of rounds ends when the surplus of all the buyers reduces to zero. At the end, we readjust the allocations to merge the two price tiers. In what follows, we explain our algorithm in three parts: a) Construction and properties of the demand graph, b) Transfer walks and c) Readjustment of allocations.

\textbf{Notation}: We denote the allocations of good $j$ to buyer $i$ at prices $p_j$ and $p_j/(1+\epsilon)$ as $h_{ij}$ and $y_{ij}$ respectively.
We denote by $z_j = 1 - \sum_{i}{(h_{ij} + y_{ij})}$ the amount of good $j$ unassigned at any point in the algorithm. Given any prices and allocations, the surplus $r_i$ of buyer $i$ is the part of her budget unspent: $$r_i\ =\ B_i\ -\ \sum_{j}{(p_j+c_{ij})h_{ij}}\ -\ \sum_j{\left(\frac{p_j}{1+\epsilon}+c_{ij}\right)y_{ij}}$$

Notice that since the prices remain constant throughout a round except at the end, the demand sets of all the buyers are well defined. In each round we fix a function $\pi(i) = \min\{\, j\, |\, j \in D_i\,\}$. Intuitively, we will attempt to allocate the good $\pi(i)$ to $i$ in this round, ignoring all the other goods in $D_i$ for the moment. Any choice of a good from $D_i$ suits as $\pi(i)$, but we fix a function for ease of exposition.

\vspace{7mm}

\noindent\textbf{Construction and properties of the demand graph}\\

We then construct a directed graph $G$ on the set of buyers. An edge exists from buyer $i$ to $k$ if and only if  $y_{k\pi(i)} > 0$. A node $i$ in this graph with (1) no out-edges (\textit{i.e.} a sink), (2) $r_i > 0$ and (3) $z_{\pi(i)} = 0$ will be defined to be `unsatisfiable'.

\begin{lemma}\label{lem.priceincrease}
For an unsatisfiable node $i$, the price of the good $\pi(i)$ can be increased by a multiplicative factor of $1+\epsilon$.
\end{lemma}
\begin{proof}

Let $j = \pi(i)$. Since node $i$ is unsatisfiable, all the allocations of good $j$ are at the current prices before the price raise and they shift to the lower price tier after the price raise. This maintains invariant I5. The fact that all these allocations continue to be charged at the earlier price maintains invariant I1. Invariant I3 follows from $z_j = 0$. We now need to verify that invariant I4 is not violated.

Let $k$ be any buyer such that $h_{kj} > 0$. We claim that $j \in D_k$ before the price raise. For contradiction, assume otherwise. Then $u_{kj} < \alpha_k(p_j+c_{kj})$. Now consider the last instance in the algorithm when any allocation of good $j$ was made to buyer $k$. If $\alpha'_k$ was the bang-per-buck of $k$ at that instance, then $u_{kj} = \alpha'_k(p_j+c_{kj})$ implying $\alpha_k > \alpha'_k$. This is impossible since by definition of bang-per-buck, the $\alpha$ values decrease monotonically as the prices are raised.

Hence we have $u_{kj} = \alpha_k(p_j+c_{ij})$. Let $p'_j = (1+\epsilon)p_j$ and $\alpha'_k$ be the bang-per-buck of $k$ after the price increase. Then $$u_{kj}\ =\ \alpha_k(p_j+c_{ij})\ \geq\ \alpha'_k\cdot\left(\frac{p'_j}{1+\epsilon}+c_{ij}\right)\ \geq\ \alpha'_k\cdot\frac{(p'_j+c_{ij})}{1+\epsilon}$$ which is exactly the statement of invariant I4.

Note: The existence of buyer $i$ with $r_i > 0$ is not essential to the statement of the proof, but is rather an artefact of the algorithm. Intuitively, we only raise the price of a good if the current price leads to excess demand, as evident from buyer $i$.
\end{proof}

But the graph $G$ may not contain an unsatisfiable node to start with. Hence we perform a series of reallocations until we create and/or find such a node.

The reallocation involves the following step: For an edge $i \rightarrow k$ in $G$ with $r_i > 0$, we take away the lower price allocation of good $\pi(i)$ for $k$ and allocate it to $i$ at the current price. In short, we perform the operations $y_{k\pi(i)} \leftarrow y_{k\pi(i)} - \delta$ and $h_{i\pi(i)} \leftarrow h_{i\pi(i)} + \delta$ for a suitably chosen value of $\delta$. This process reduces $r_i$, $y_{k\pi(i)}$ and increases $r_k$. If $y_{k\pi(i)}$ reduces to zero, we drop the edge $(i,k)$ from the graph. When we make such a reallocation, we say that we
{\em transfer surplus} from $i$ to $k$.
Note that the surplus is not conserved.
This is because the price paid by $i$ for the same amount of the good,
including the transaction costs, could even be lower than the price paid by $k$.

\begin{lemma}\label{lem.transfer}
If the edge from $i$ to $k$ exists in $G$ with $r_i > 0$,
then we can transfer surplus from $i$ to $k$ such that either the
surplus of $i$ becomes zero or the edge $(i,k)$ drops out of $G$.
\end{lemma}
\begin{proof}
Refer to Section \ref{sec:analysis}.
\end{proof}

We can repeatedly apply lemma \ref{lem.transfer} to transfer surplus along a path in $G$.

\begin{corollary}\label{lem.path}
If there exists a path from $i$ to $k$ in $G$ and $r_i>0$,
then we can transfer surplus from $i$ to $k$ such that either the
surplus of all the nodes on the path except $k$ becomes zero or an edge in the path drops out of $G$.
\end{corollary}
\begin{proof} Refer to Section \ref{sec:analysis}.
\end{proof}

Finally, $G$ may contain cycles. Consider the edges $(i_1, i_2)$ and $(i_2, i_3)$ in $G$ and let $j_1 = \pi(i_1)$ and $j_2 = \pi(i_2)$. If the transaction costs are all zero, then it can be argued that the last price raise for $j_1$ must have taken place before the last price raise for $j_2$. Telescoping this argument, one can preclude the existence of cycles in $G$ in absence of transaction costs. This acyclicity of $G$ forms a pivotal argument in the algorithm of Garg and Kapoor \cite{GK04}. In Appendix \ref{app3} we provide a sketch of how a cycle can emerge in $G$ when transaction costs are present. Moreover, we can show that the algorithm of \cite{GK04} will slow down indefinitely if $G$ contains cycles. In the above example suppose $i_1 = i_3$, so that $G$ contains a cycle of length two. For simplicity, assume $r_{i_1} > 0$ and $r_{i_2} = 0$. One round of their algorithm may perform a surplus transfer from $i_1$ to $i_2$ followed by another from $i_2$ to $i_1$. One can compute that this will change $r_{i_1}$ as $r_{i_1} \leftarrow Cr_{i_1}$ where $C$ is a constant at given prices. If $C = 1-\delta$ for a very small $\delta > 0$, it can require an infinite number of rounds for $r_{i_1}$ to reduce to zero. Broadly speaking, any distributed version of the auction algorithm, that is unaware of the graph structure will suffer from the same problem.

Therefore, we need to be able to transfer surplus around a cycle.

\begin{lemma}\label{lem.cycle}
If there exists a cycle in $G$ and exactly one node in the cycle has positive surplus, then we can transfer surpluses in such a way that either all the node in the cycle have zero surplus or an edge in the cycle drops out.
\end{lemma}
\begin{proof} Refer to Section \ref{sec:analysis}
\end{proof}

In a round, we use the above lemmas to perform multiple iterations of the transfer walk.

\medskip

\noindent \textbf{Transfer Walk}

\indent\textbf{Step 1}: Find a node $i_0$ with a positive surplus. If there are no such nodes, then terminate the round and jump to readjustment of allocations.

\indent\textbf{Step 2}: Follow a path going out of $i_0$ in $G$ in a depth-first-search fashion. We look at the first edge in the adjacency list of the last visited node $i$ on the path.  Let $(i,k)$ be this edge. If node $k$ is yet unvisited, we follow that edge to extend the path. If $k$ is already on the path, then we have found a cycle in $G$. Finally if $i$ has no out-edges, then we have found a sink.
Whichever the case, we now transfer surplus along the current path from $i_0$ to $i$ as in corollary \ref{lem.path}. If an edge along the path drops out, we trigger event 2d. Otherwise, we trigger events 2a-2c depending upon case. Since the path can visit at most $n$ new nodes, the transfer walk must end in a finite number of operations in one the of following events:
    \begin{itemize}
    	\item[]\textbf{Event 2a} - The path reaches a sink $i$ with $z_{\pi(i)} = 0$: Let $j = \pi(i)$. By corollary \ref{lem.path}, we must have transferred a positive surplus to $i$ even if $r_i$ was zero at the begining of the walk. Hence $i$ is an unsatisfiable node. Raise $p_j \leftarrow (1+\epsilon)p_j$. Terminate the walk and the round.
    	
    	\item[]\textbf{Event 2b} - The path reaches a sink $i$ with $z_{\pi(i)} > 0$: Let $j = \pi(i)$. By invariant I3, $p_j = \epsilon$. We let $\delta = \min(\ r_i/\epsilon,\ z_j\ )$. We then assign $h_{ij}\ \leftarrow\ h_{ij} + \delta$. If $\delta\ =\ r_i/\epsilon$ then the surplus of $i$ goes to zero otherwise $z_j$ goes to zero. In either case we end the this transfer walk.
    	
    	\item[]\textbf{Event 2c} - The path finds a cycle: Let $i$ be the last node visited on the path and an edge $(i,k)$ in $G$ reaches a node $k$ already visited on the path. By corollary \ref{lem.path}, all the nodes in the cycle except $i$ have zero surplus. Therefore, we apply lemma \ref{lem.cycle} until the surplus of $i$ becomes zero or an edge in the cycle drops out. We terminate the current walk.
    	
    	\item[]\textbf{Event 2d} - An edge drops out during path transfer: In this case we terminate the current walk.
	\end{itemize}

If a transfer walk ends in event 2a, we terminate the current round and start the next one. Otherwise if events 2b-2d are triggered, we start a new transfer walk. If the surplus of all buyers is found to be zero in Step 1, we move to the last phase, which is readjustment of allocations.

\subsubsection*{Readjustment of allocations}

At the end of the transfer walks, all the required invariants are satisfied, but the same good may be allocated to the same or different buyers at different prices: $p_j$ and $p_j/(1+\epsilon)$. Therefore in this phase, we merge the two tiers of allocation for every buyer-good pair to create the final allocations. For all $i$, $j$ such that $y_{ij} > 0$, we assign $$x_{ij}\ \leftarrow\ h_{ij}\ +\ \frac{\frac{p_j}{1+\epsilon} + c_{ij}}{p_j+c_{ij}}y_{ij}$$ The final equilibrium prices are the prices at the termination of the algorithm.

\begin{theorem}
\label{thm:correct}
The algorithm produces $\epsilon$-approximate equilibrium prices and allocations.
\end{theorem}
\begin{proof} 
By lemmas \ref{lem.priceincrease}, \ref{lem.transfer}, \ref{lem.path} and \ref{lem.cycle}, the prices and allocations at the end of the last round satisfy invariants I1-I5. Invariant I4 implies that the final allocations satisfy the approximate optimality dictated by equation \eqref{eq4}. We have for all $i$,
$$\sum_{j}{(p_j+c_{ij})x_{ij}}\ =\ \sum_{j}{(p_j+c_{ij}) \left(h_{ij} + \frac{y_{ij}}{1+\epsilon}\right)}\ =\ \sum_{j}{(p_j+c_{ij})h_{ij}}\ +\ \sum_{j}{\left(\frac{p_j}{1+\epsilon}+c_{ij}\right)y_{ij}}\ =\ B_i$$
which proves equation \eqref{eq5}. Equation \eqref{eq6} follows by definitions of $x_{ij}$'s. Finally, we need to show that equation \eqref{eq3} holds. First observe that $p_j > \epsilon$ implies that the price of good $j$ was raised in some round, which implies $z_j = 0$ in that round. Invariant 3 then implies that $z_j = 0$ at the end of the last round. From definition of $x_{ij}$, $x_{ij}\ \geq\ h_{ij} + y_{ij}/(1+\epsilon)\ \geq\ (h_{ij} + y_{ij})/(1+\epsilon)$ for all $i$ and $j$. Summing over $i$, this gives equation \eqref{eq3}.
\end{proof}

\section{Analysis}
\label{sec:analysis}
\subsection*{Correctness of the Algorithm}
We will now prove the lemmas used in the preceding section.
\medskip

\noindent{\bf Proof of Lemma \ref{lem.transfer}:}
We choose the amount of good $j = \pi(i)$ transferred from $k$ to $i$ as 
$\delta = \max\left(\frac{r_i}{p_j + c_{ij}},\ y_{kj}\right).$ 
Clearly, if $\delta = y_{kj}$ then the lower price allocation of $j$ to $k$ is exhausted and the edge $(i,k)$ drops out. Otherwise, buyer $i$ spends all her surplus on the new allocation.
\medskip

\noindent{\bf Proof of Corollary \ref{lem.path}:}
Let $i = v_0, v_1, ...., v_l  = k$ be the path. Then use lemma \ref{lem.transfer} to transfer surplus from $v_q$ to $v_{q+1}$ for $q$ = $0$ to $l-1$ in that order. If no edge in the path drops out of $G$, then by the lemma, the surplus of all the nodes on the path except $v_l = k$ goes to zero.
\medskip

\noindent{\bf Proof of Lemma \ref{lem.cycle}:}
Let $v_0, v_1, ...., v_l$ be the cycle in $G$ with $v_0 = v_l$ and w.l.o.g. $v_0$ being the node with positive surplus. We will transfer a quantity $\delta_q$ of good $\pi(v_q)$ from the lower price allocation of $v_{q+1}$ to the higher price allocation of $v_{q}$. We will adjust the $\delta$ values carefully so as to maintain zero surplus at all nodes except $v_0$.

Given $\delta_0$, the above requirement fixes all other $\delta$ values. To maintain zero surpluses at all other nodes, we need to set: $$\delta_{q+1}\ =\ \frac{1}{p_{\pi(v_{q+1})} + c_{v_{q+1}\pi(v_{q+1})}}\cdot\left(\frac{p_{\pi(v_q)}}{1+\epsilon} + c_{v_{q+1}\pi(v_{q})}\right)\cdot\delta_q$$

Note that this process changes the surplus of $v_0$ as $$r_{v_0}\ \ \leftarrow\ \ r_{v_0}\ +\ \left(\frac{p_{\pi(v_{l-1})}}{1+\epsilon} + c_{v_{0}\pi(v_{l-1})}\right)\delta_{l-1}\ -\ (p_{\pi(v_{0})} + c_{v_{0}\pi(v_{0})})\delta_{0}\ \ =\ \ r_{v_0} + C\delta_0$$ where $C$ is a constant determined by current prices. If $C$ is negative, then $r_{v_0}$ reduces and let $\delta^*$ be the value of $\delta_0$ at which is reaches zero. Otherwise, let $\delta^* = \infty$.

We set $\delta_0$ to be the maximum value such that $\delta_0 \leq \delta^*$ and $\delta_q \leq y_{v_{q+1}\pi(v_q)}$ for all $0 \leq q < l$. We perform the surplus transfer with these $\delta$ values. If $\delta_0 = \delta^*$ then the surplus of $v_0$ reduces to zero and hence all buyers have zero surpluses. Otherwise, there exists $q$ such that $\delta_q = y_{v_{q+1}\pi(v_q)}$ and hence the edge $(v_q, v_{q+1})$ drops out of the graph.

\subsection*{Running Time of the Algorithm}

The major chunk of the computation in our algorithm happens inside the transfer walks. Hence we count the number of transfer walks that we perform, categorized by the event that ends the walk.

\begin{lemma}
\label{lem:edgedrop}
If $R$ is the number of rounds in the algorithm, then the number of transfer walks that end in an edge dropping out of $G$ is at most $nR$.
\end{lemma}
\begin{proof}
There are two ways in which edges can be added to $G$ after a price increase.
\begin{enumerate}
\item After a price increase, the function $\pi(i)$ can change and hence the edges going out of $i$ can change.
\item When the price of a good $j$ is raised at the end of the round, an edge $(i, k)$ appears in $G$ for each $i$ such that $j \in D_i$ and for each $k$ which has any allocation of $j$.
\end{enumerate}
This addition of potentially $\Omega(n^2)$ edges implies a weak upper bound of $n^2R$ on the number of such transfer walks in the algorithm. In what follows we define another graph which mirrors $G$ in semantics, but contains a lot fewer edges.

Let $H$ be a directed bipartite graph between the set of buyers $I$ and the set of goods $J$. For $i \in I$ and $j \in J$, an edge $(i, j) \in H$ if and only if $j = \pi(i)$. Similarly, $(j,i) \in H$ if and only if $y_{ij} > 0$. Clearly, the edge $(i, k) \in G$ if and only if $k$ can be reached from $i$ in $H$ by a path of length two. Also note that the number of edges going from $I$ to $J$ is exactly $n$ at any point. An edge $(i,k)$ drops out of $G$ if and only if the edge $(\pi(i), k)$ drops out of $H$. But only $n$ edges are possibly added from $J$ to $I$ after a price increase. Since we start with zero edges going from $J$ to $I$, an edge can drop out of $H$  at most $nR$ times throughout the algorithm, and hence the same bound applies to $G$.
\end{proof}

\medskip

\noindent\textbf{Proof of Theorem \ref{thm.main}}:

\textbf{Initialization and readjustment}: Both the initialization and final adjustment of allocations can be performed in $mn$ operations.

\textbf{The number of rounds}: The price of exactly one good is raised by multiplicative factor of $1+\epsilon$ in each round except the last round. Starting at $\epsilon$, the maximum value to which a price may be raised is $B = (1+\epsilon)\sum_{i}B_i$. Therefore, there can be at most $R = 1+\frac{m}{\epsilon}\log(\frac{B}{\epsilon})$.

\textbf{Constructing the graph}: Notice that although the demand set of a buyer may contain all the $m$ goods, we only need one of them at any point. For each buyer, we maintain all the goods in a balanced tree data structure that sorts the goods first by the bang-per-buck $u_{ij}/(p_j + c_{ij})$ and then by the index $j$. In this manner, we can compute the function $\pi(i)$ in $O(\log{m})$ time. Given $\pi(i)$, every node may have an edge to every other node. Therefore, the graph $G$ can be constructed in $O(n^2 + n\log{m})$ operations. After the price increase at the end of the round, the sorted trees can be maintained in time $O(n\log{m})$ while the transfer of allocations from higher to lower price tier can be completed in $O(n)$ operations.

\textbf{Number of transfer walks}: All the remaining computation in the algorithm takes place within the transfer walks. We will perform an amortized analysis on the number of transfer walks that take place throughout the algorithm. Notice that since we follow the first edge going out of each vertex, the depth-first-search requires only $O(n)$ operations. The surplus transfer along a path and a cycle can similarly be performed in $O(n)$ operations. When an edge drops out, updating $G$ involves simply incrementing a pointer. Therefore, overall a transfer walk requires $O(n)$ operations.

We will now bound the number of transfer walks that happen throughout the algorithm, including all the rounds. We will classify them by the event that ends the walk. At most $R$ transfer walks can terminate the round. At most $m$ walks can end with $z_j$ going zero. Lemma \ref{lem:edgedrop} bounds the number of walks that end with an edge dropping out of the graph. The only remaining case is that the walk ends when the surplus of the last visited node on the path vanishes. A transfer walk ending in this case leaves one less node in $G$ with a positive surplus. To see this, observe that a transfer walk starts with a node on the same path with positive surplus and by the time it ends in this case, all the nodes on the path have zero surplus by corollary \ref{lem.path} and lemma \ref{lem.cycle}.

Let $r_+$ be the number of nodes in $G$ with positive surplus at any point in the algorithm. After initialization we have $r_+ = n$ and the only event which may increase $r_+$ is event 2d. If an edge $(i,k)$ drops out during surplus transfer along the path, node $k$ may be left with some positive surplus that was absent at the start of the walk. Therefore $r_+$ increases by at most one in this event. Combined with lemma \ref{lem:edgedrop}, this implies a bound of $n+nR$ on the number of times $r_+$ reduces.

It is clear from the above analysis that the algorithm performs at most $O(nR)$ transfer walks. Combined with the other computation bounds, this yields an upper bound of $O\left(\frac{1}{\epsilon}(n+\log{m})mn\log(B/\epsilon)\right)$ on the running time of the algorithm.

\bibliography{mewtc}
\bibliographystyle{plain}

\medskip

\medskip

\medskip

\appendix
\section{Formulation as a Convex Program}
\label{convex}

\noindent\textbf{Convex Program Formulation}: Consider the following convex program,
\begin{eqnarray}
\mathrm{minimize} && \sum_j p_j -\sum _i B_i \log \beta_i \label{eqn.cp}\\
\forall i,j & ~~~~~~ &  p_j + \cij\ \geq\ \uij \beta_i\nonumber\\
\forall i,~~\beta_i\ \geq\ 0 & & \forall j,~~p_j\ \geq\ 0\nonumber
\end{eqnarray}
\begin{theorem}
A vector $(\p, \bf\beta)$ is an equilibrium of $\market$ if and only if it minimizes Convex program \eqref{eqn.cp}.
As a corollary, a Fisher market with linear utilities and transaction costs has a unique set of equilibrium prices.
\end{theorem}

\begin{proof}

Note that the KKT conditions that guarantee optimality of a feasible
solution are as follows,
with $\xij$ being the Lagrangian multiplier for the inequality $p_j +
\cij \geq \uij \beta_i$.
There exists, for all $i,j,$  $\xij$ such that
\begin{eqnarray}
&& \forall~i,j: \ \xij > 0 \Rightarrow p_j + \cij = \uij \beta_i\label{eqn.cs1}\\
&& \forall~i:\ \ \sum_j \uij \xij = B_i /\beta_i.\label{eqn.cs3}\\
&& \forall~j:\ \ \sum_i\xij \leq  1 ~~\text{ and equality holds if} ~~
p_j >0.\label{eqn.cs2}
\end{eqnarray}

The equivalence of the KKT conditions and market equilibrium conditions (equations \eqref{eq7}-\eqref{eq2}) follows
from interpreting $\xij$
as the allocation and $\beta_i$ as $1/\alpha_i$.
Conditions (\ref{eqn.cs1}) and (\ref{eqn.cs3})
together imply \eqref{eq7}, that
all buyers exhaust their budget, i.e. for all $i$,
\[ \sum_j (p_j + \cij) \xij\ =\ \sum_{j:\ x_{ij}>0} \uij \beta_i \xij\ =\ B_i.\]

Conversely, equations \eqref{eq7} and \eqref{eq2} together imply \eqref{eqn.cs3}. $$\sum_{j}{u_{ij}x_{ij}}\ =\ \sum_{j:\ x_{ij}>0}{\frac{(p_j+c_{ij})x_{ij}}{\beta_i}}\ =\ \frac{B_i}{\beta_i}$$

Conditions \eqref{eq8} and \eqref{eq1}  are  the same as (\ref{eqn.cs2}).
Condition \eqref{eq2} is  the same as (\ref{eqn.cs1}).
This proves existence of equilibrium prices. 

Uniqueness follows from the fact that the objective function in
\eqref{eqn.cp} is strictly convex. Let $(\p, \beta)$ and $(\overline{\p}, \overline{\beta})$ be two distinct optimal solutions to the convex program. We claim that $\beta \neq \overline{\beta}$. For contradiction, assume $\beta = \overline{\beta}$. But observe that given any $\beta$, we can determine a \emph{unique} $\p$ which minimizes the objective function by setting $p_j = \max_i{(u_{ij}\beta_i - c_{ij})}$. Therefore, if $\beta = \overline{\beta}$ then $\p = \overline{\p}$. Since $(\p, \beta)$ and $(\overline{\p}, \overline{\beta})$ are distinct, we conclude that $\beta \neq \overline{\beta}$.

Now consider a linear combination $(\hat{\p}, \hat{\beta}) = \alpha(\p, \beta) + (1-\alpha)(\overline{\p}, \overline{\beta})$ for any $0 < \alpha < 1$. Clearly, $\sum_j{\hat{p}_j} = \sum_j{\left(\alpha p_j + (1-\alpha)\overline{p}_j\right)}$, but $\sum_i{B_i\log{\hat{\beta}_i}} > \sum_i{B_i\left(\alpha\log{ \beta_i} + (1-\alpha)\log{\overline{\beta}_i}\right)}$, since $\log$ is a strictly concave function. Therefore, $$\sum_j{\hat{p}_j} - \sum_i{B_i\log{\hat{\beta}_i}}\ <\ \alpha\left(\sum_j{p_j} - \sum_i{B_i\log{\beta_i}}\right) + (1-\alpha)\left(\sum_j{\overline{p}_j} - \sum_i{B_i\log{\overline{\beta}_i}}\right)$$

This contradicts the fact that $(\p, \beta)$ and $(\overline{\p}, \overline{\beta})$ were both optimal. Hence, the objective function is strictly convex, and the equilibrium is unique.
\end{proof}

\section{A market with irrational equilibrium prices}
\label{app1}
Consider a market with two buyers, $i$ and $k$ and two goods $j$, $j'$.

\medskip

\begin{tabular}{| c l  c | c c c | c c c |}
\hline
& & & & Buyer $i$ & & & Buyer $k$ & \\
\hline
& Budgets & & & 1 & & & 1 & \\
\hline
& Utilities & & & $u_{ij} = 1000$, $u_{ij'} = 1$ & & & $u_{kj} = u_{kj'} = 1$ & \\
\hline
& Transaction costs & & &  $c_{ij} = 1$, $c_{ij'} = 1000$ & & & $c_{kj} = c_{kj'} = 0$ & \\
\hline
\end{tabular}

\medskip

It can be verified that the $p_j = p_{j'} = \frac{1}{\sqrt{2}}$ are equilibrium prices. At these prices, buyer $i$ only demands good $j$ whereas buyer $i'$ demands both goods. Buyer $i$ spends her one dollar on $\frac{\sqrt{2}}{\sqrt{2} + 1}$ units of good $j$ at effective price $1 + \frac{1}{\sqrt{2}}$. Buyer $i'$ buys the remaining $\frac{1}{\sqrt{2}+1}$ units of good $j$ and the entire one unit of good $j'$, both at effective price $\frac{1}{\sqrt{2}}$.

\section{Existence of cycles in the demand graph}
\label{app3}

We provide a sketch of how a cycle can exist in the demand graph. Consider two buyers $i$ and $k$ and two goods $j$ and $j'$ in a market with many other buyers and goods. Their utilities and transaction costs are tabulated below. Assume both $i$ and $j$ have sufficiently large budgets and consider an instance in the auction algorithm when $p_j = p_{j'} = 1$.

\medskip

\begin{tabular}{| c l  c | c c c | c c c |}
\hline
& & & & Buyer $i$ & & & Buyer $k$ & \\
\hline
& Utilities & & & $u_{ij} = 1+\frac{\epsilon}{3}$, $u_{ij'} = 2$ & & & $u_{kj} = 2$, $u_{kj'} = 1+\frac{\epsilon}{3}$ & \\
\hline
& Transaction costs & & &  $c_{ij} = 0$, $c_{ij'} = 1$ & & & $c_{kj} = 1$, $c_{kj'} = 0$ & \\
\hline
\end{tabular}
\ \ where $\epsilon < 1$.

\medskip

Clearly, good $j$ will be allocated to buyer $i$ and good $j'$ to buyer $k$. Now if the price of both the goods is raised due to demand by other buyers, $j'$ appears in the demand set of $i$ and $j$ appears in the demand set of $k$. Hence, the graph $G$ may contain a cycle on the two nodes $i$ and $k$ after the prices are raised.

Note: The fast version of the auction algorithm in \cite{GK04} employs certain tie-breaking techniques that suffice to ensure acyclicity in the absence of transaction costs. The above example holds even if those techniques are employed in our model. Therefore, we have chosen to simplify our algorithm by not using these tie-breaking rules.

In absence of transaction costs the following property holds, which is used in \cite{GK04}: If buyer $i$ prefers good $j$ over good $j'$ at prices $p_j$ and $p_{j'}$, then she maintains the same preference when the prices are each raised to $p_j(1+\epsilon)$ and $p_{j'}(1+\epsilon)$. From the above example, such an assertion is false in our model.

\end{document}